\documentclass[10pt,conference]{IEEEtran}
\usepackage{fixltx2e}
\usepackage{cite}
\usepackage{url}
\usepackage{mathrsfs}
\usepackage{color}
\usepackage{float}
\usepackage[caption=false]{subfig}

\ifCLASSINFOpdf
   \usepackage[pdftex]{graphicx}
   \graphicspath{{Figs/}}
   \DeclareGraphicsExtensions{.pdf,.jpeg,.png}
\else
\fi

\usepackage[cmex10]{amsmath}
\usepackage{amsmath}
\usepackage{amssymb}
\usepackage{amsthm}
\usepackage{amsfonts}
\usepackage{bm}
\usepackage{xfrac}
\usepackage{empheq}
\usepackage[normalem]{ulem} 
\usepackage{soul} 
\usepackage{mathtools}

\newtheorem{theorem}{Theorem}

\newtheorem{example}{Example}
\newtheorem{proposition}{Proposition}
\newtheorem{lemma}{Lemma}

\newtheorem{corollary}{Corollary}
\newtheorem{remark}{Remark}
\theoremstyle{definition}
\newtheorem{definition}{Definition}

\usepackage[a4paper,bindingoffset=0.2in,%
left=1in,right=1in,top=1in,bottom=1in,%
footskip=.25in]{geometry}
\graphicspath{{figs/}}

\begin{document}
	\newgeometry{left=0.7in,right=0.7in,top=.5in,bottom=1in}
	\title{Multi-User Privacy Mechanism Design with Non-zero Leakage}
\vspace{-5mm}
\author{
		\IEEEauthorblockN{Amirreza Zamani, Tobias J. Oechtering, Mikael Skoglund \vspace*{0.5em}
			\IEEEauthorblockA{\\
                              Division of Information Science and Engineering, KTH Royal Institute of Technology \\
				Email: \protect amizam@kth.se, oech@kth.se, skoglund@kth.se }}
		}
	\maketitle
%
\begin{abstract}
	A privacy mechanism design problem is studied through the lens of information theory. 
In this work, an agent observes useful data $Y=(Y_1,...,Y_N)$ that is correlated with private data $X=(X_1,...,X_N)$ which is assumed to be also accessible by the agent. Here, we consider $K$ users where user $i$ demands a sub-vector of $Y$, denoted by $C_{i}$. The agent wishes to disclose $C_{i}$ to user $i$. Since $C_{i}$ is correlated with $X$ it can not be disclosed directly. 
A privacy mechanism is designed to generate disclosed data $U$ which maximizes a linear combinations of the users utilities while satisfying a bounded privacy constraint in terms of mutual information. In a similar work it has been assumed that $X_i$ is a deterministic function of $Y_i$, however in this work we let $X_i$ and $Y_i$ be arbitrarily correlated. 

First, an upper bound on the privacy-utility trade-off is obtained by using a specific transformation, Functional Representation Lemma and Strong Functional Representaion Lemma, then we show that the upper bound can be decomposed into $N$ parallel problems. 
Next, lower bounds on privacy-utility trade-off are derived using Functional Representation Lemma and Strong Functional Representaion Lemma. The upper bound is tight within a constant and the lower bounds assert that the disclosed data is independent of all $\{X_j\}_{i=1}^N$ except one which we allocate the maximum allowed leakage to it. 
Finally, the obtained bounds are studied in special cases.
\end{abstract}
\section{Introduction}

Recently, the privacy mechanism design problem through the lens of information theory is receiving increased attention
\cite{Calmon2, yamamoto, issa, makhdoumi, sankar,borz, gun,khodam,Khodam22,kostala, dwork1, calmon4, issajoon, asoo, Total, issa2, zamani2022bounds, zamaniITW2022, kosenaz, naz2, 9457633,asoodeh1,MMSE,nekouei2,bassi}. 
Specifically, 
fundamental limits of the privacy utility trade-off measuring the leakage using estimation-theoretic guarantees are studied in \cite{Calmon2}. A related source coding problem with secrecy is studied in \cite{yamamoto}.
The concept of maximal leakage has been introduced in \cite{issa} and some bounds on the privacy utility trade-off have been derived. 
The concept of privacy funnel is introduced in \cite{makhdoumi}, where the privacy utility trade-off has been studied considering the log-loss as privacy measure and a distortion measure for utility. 

The privacy-utility trade-offs considering equivocation and expected distortion as measures of privacy and utility are studied in both \cite{yamamoto} and \cite{sankar}.
In \cite{borz}, the problem of privacy-utility trade-off considering mutual information both as measures of utility and privacy given the Markov chain $X-Y-U$ is studied. It is shown that under the perfect privacy assumption, i.e., no leakages are allowed, the privacy mechanism design problem can be obtained by a linear program. This has been extended in \cite{gun} considering the privacy utility trade-off with a rate constraint on the disclosed data.
\begin{figure}[]
	\centering
	\includegraphics[scale = .11]{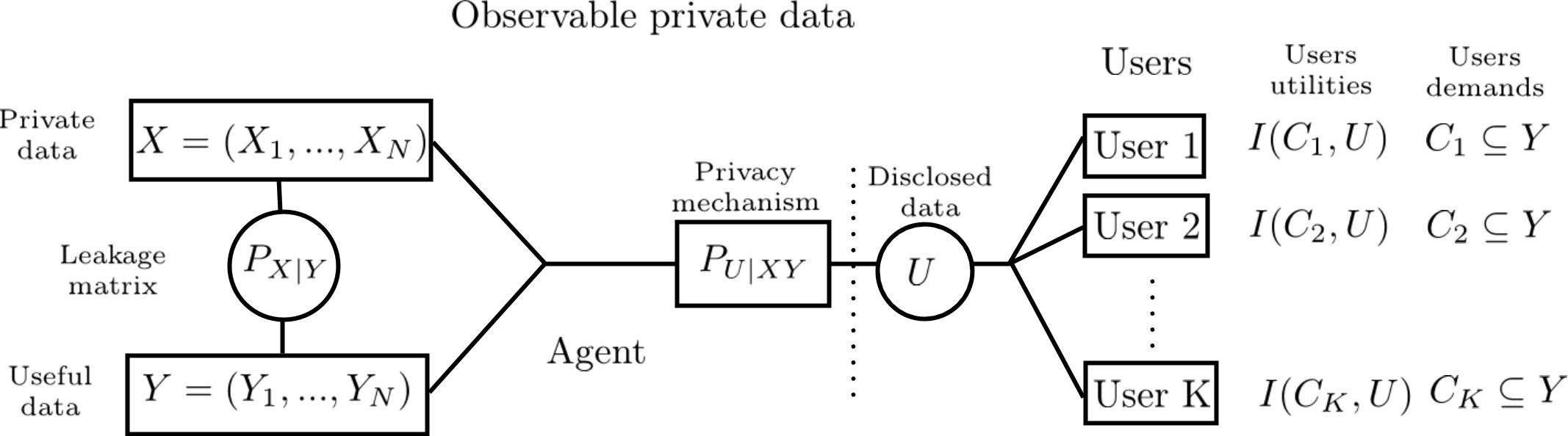}
	\caption{In this work the agent has access to $Y$ and $X$. Each user demands a sub-vector of $Y$ and the agent releases message $U$ which maximizes a linear combination of utilities.}
	\label{ITWsys}
\end{figure}
Moreover, in \cite{borz}, it has been shown that information can be only revealed if the kernel (leakage matrix) between useful data and private data is not invertible. In \cite{khodam}, we generalize \cite{borz} by relaxing the perfect privacy assumption allowing some small bounded leakage. More specifically, we design privacy mechanisms with a per-letter privacy criterion considering an invertible kernel where a small leakage is allowed. We generalized this result to a non-invertible leakage matrix in \cite{Khodam22}.\\
In \cite{kostala}, the problem of \emph{secrecy by design} is studied and bounds on privacy-utility trade-off for two scenarios where the private data is hidden or observable are derived by using the Functional Representation Lemma. These results are derived under the perfect secrecy assumption, i.e., no leakages are allowed. 
In \cite{zamani2022bounds}, we generalized the privacy problems considered in \cite{kostala} by relaxing the perfect secrecy constraint and allowing some leakages. 
Furthermore, in the special case of perfect privacy we derived a new upper bound for the perfect privacy function and it has been shown that this new bound generalizes the bound in \cite{kostala}. Moreover, it has been shown that the bound is tight when $|\mathcal{X}|=2$.\\
In \cite{zamaniITW2022}, we considered the privacy-utility trade-off with two different per-letter privacy constraints in two scenarios where the private data is hidden or observable. Upper and lower bounds are derived and it has been shown that the bounds in the first scenario, where the private data is hidden to the agent, are asymptotically optimal when the private data is a deterministic function of useful data. \\
Our problem here is closely related to \cite{9457633}, where fundamental limits of private data disclosure are studied. The goal is to minimize the leakage under the utility constraints with non-specific tasks. It has been shown that under the assumption that the private data is an element-wise deterministic function of useful data, 
the main problem can be reduced to multiple privacy funnel (PF) problems. Moreover, the exact solution to each problem has been obtained.  \\
In this paper, $Y=(Y_1,...,Y_N)$ denotes the useful data where $Y_1,...,Y_N$ are mutually independent random variables (RV). The useful data is correlated with the private data denoted by $X=(X_1,...,X_N)$ where $X_1,...,X_N$ are mutually independent RVs. As shown in Fig.~\ref{ITWsys}, user $i$, $i\in\{1,...,K\}$, demands an arbitrary sub-vector of $Y$ denoted by $C_i$ and an agent wishes to design disclosed data denoted by $U$ which maximizes a linear combination of the utilities (weighted sum-utility) while satisfying the bounded leakage constraint, i.e., $I(X;U)\leq \epsilon$. Utility of user $i$ is measured by the mutual information between $C_i$ and $U$, i.e., $I(C_i,U)$. The problem considered in this paper is motivated by the dual problem of the privacy design problem studied in \cite{9457633}. The assumption that the private data is a deterministic function of useful data is very restrictive, hence, in this work we generalize the assumption by letting $X_i$ and $Y_i$ be arbitrarily correlated.     

We first find upper bounds on the privacy-utility trade-off by using the same transformation as used in \cite[Th. 1]{9457633} and \cite[Th. 1]{liu2020robust}. We then show that the upper bound can be decomposed into $N$ parallel privacy design problems. Moreover, lower bounds on privacy-utility trade-off are obtained by using Functional Representation Lemma (FRL) and Strong Functional Representation Lemma (SFRL). The lower bound provides us a privacy mechanism design which is based on the randomization used in \cite{zamani2022bounds}. In \cite{zamani2022bounds}, we used randomization to extend FRL and SFRL relaxing the independence condition. 
Furthermore, a similar randomization technique has been used in \cite{warner1965randomized}. We show that the upper bound is tight within a constant term and the lower bound is optimal when the private data is an element-wise deterministic function of the useful data, i.e., the mechanism design is optimal in this case. 
Finally, we study the bounds in different scenarios.

\section{system model and Problem Formulation} \label{sec:system}
Let $P_{XY}(x,y)=P_{X_1,..,X_N,Y_1,..,Y_N}(x_1,..,x_N,y_1,..,y_N)$ denote the joint distribution of discrete random variables $X$ and $Y$ defined on finite alphabets $\cal{X}$ and $\cal{Y}$. Here, we assume that $\{(X_i,Y_i)\}_{i=1}^N$ are mutually independent, hence, $P_{XY}(x,y)=\prod_{i=1}^N P_{X_i,Y_i}(x_i,y_i)$, where $X_i$ and $Y_i$ are arbitrarily correlated. The demand of user $i$, i.e., $C_i$, $i\in{1,...,K}$, is a discrete random variable with finite alphabet $\mathcal{C}_i$. As we mentioned earlier the agent has access to both $X$ and $Y$ and design disclosed data $U$ defined on $\mathcal{U}$ which maximizes a linear combination of the utilities while satisfying a privacy leakage constraint.  
The relation between $X$ and $Y$ is given by $P_{X|Y}$ where we represent the leakage matrix $P_{X|Y}$ by a matrix defined on $\mathbb{R}^{|\mathcal{X}|\times|\cal{Y}|}$. 
The relation between $U$ and the pair $(Y,X)$ is described by the conditional distribution $P_{U|Y,X}(u|x,y)$.
In this work, RVs $X$ and $Y$ denote the private data and the useful data and $U$ describes the disclosed data. 


 The privacy mechanism design problem in this scenario can be stated as follows 
\begin{align}
h_{\epsilon}^K(P_{XY})&=\sup_{\begin{array}{c} 
	\substack{P_{U|Y,X}:I(X;U)\leq\epsilon}
	\end{array}}\sum_{i=1}^K \lambda_iI(C_i;U),\label{main1}
\end{align} 
where for $i\in\{1,..,K\}$, $\lambda_i\geq 0$ is fixed.\\
In the following we study the case where $0\leq\epsilon< I(X;Y)$, otherwise the optimal solution of $h_{\epsilon}^K(P_{XY})$ is $\sum_{i=1}^{K} \lambda_i H(C_i)$ achieved by $U=Y$.  
 \begin{remark}
	\normalfont
	The problem defined in \eqref{main1} is motivated by the dual to the problem considered in \cite[Eq.~(8)]{9457633}, however the Markov chain $(X,C_1,..,C_K)-Y-U$ is removed due to the assumption that the agent has access to both $X$ and $Y$ and $C_i$ is a sub-vector of $Y$. Here, instead of minimizing the leakage we maximize the linear combination of the utilities under the bounded privacy leakage constraint. 
\end{remark}  
\begin{remark}
	\normalfont
	The weights $\lambda_1,..,\lambda_K$ can correspond to the Lagrange multipliers of \cite[Eq.~(8)]{9457633}. Furthermore, they can correspond to different priorities for different users. For instance, if the utility of user $i$ is more important than the utility of user $j$ we let $\lambda_i$ be larger than $\lambda_j$.  
\end{remark}
\begin{remark}
	\normalfont 
	In case of perfect privacy, i.e., $\epsilon=0$, and $N=K=1$, \eqref{main1} leads to the secret-dependent perfect privacy function $h_0(P_{XY})$, studied in \cite{kostala}, where upper and lower bounds on $h_0(P_{XY})$ have been derived. In \cite{zamani2022bounds}, we have strengthened these bounds.
\end{remark}
 \section{Main Results}\label{sec:resul}
 In this section, we derive upper and lower bounds for $h_{\epsilon}^K(P_{XY})$. Next, we study the bounds under the perfect privacy assumption, moreover, consider the case where the private data is an element-wise deterministic function of the useful data.
 \subsection{Upper bounds on $\bm{h_{\epsilon}^K(P_{XY})}$:}
 In this section, we first derive a similar result as \cite[Th.~1]{9457633} which results in upper bounds on the privacy-utility trade-off defined in \eqref{main1}. To do so we introduce a random variable $\bar{U}=(\bar{U}_1,..,\bar{U}_N)$ which is constructed based on \cite[Th.~1]{liu2020robust} and study the corresponding properties. We then decompose the upper bound into $N$ parallel problems.
 First, let us recall the FRL \cite[Lemma~1]{kostala} and SFRL \cite[Theorem~1]{kosnane} for discrete $X$ and $Y$.
 \begin{lemma}\label{lemma1} (Functional Representation Lemma \cite[Lemma~1]{kostala}):
 	For any pair of RVs $(X,Y)$ distributed according to $P_{XY}$ supported on alphabets $\mathcal{X}$ and $\mathcal{Y}$ where $|\mathcal{X}|$ is finite and $|\mathcal{Y}|$ is finite or countably infinite, there exists a RV $U$ supported on $\mathcal{U}$ such that $X$ and $U$ are independent, i.e., we have
 	\begin{align}\label{c1}
 	I(U;X)=0,
 	\end{align}
 	$Y$ is a deterministic function of $(U,X)$, i.e., we have
 	\begin{align}
 	H(Y|U,X)=0,\label{c2}
 	\end{align}
 	and 
 	\begin{align}
 	|\mathcal{U}|\leq |\mathcal{X}|(|\mathcal{Y}|-1)+1.\label{c3}
 	\end{align}
 \end{lemma}
 \begin{lemma}\label{lemma2} (Strong Functional Representation Lemma \cite[Theorem~1]{kosnane}):
 	For any pair of RVs $(X,Y)$ distributed according to $P_{XY}$ supported on alphabets $\mathcal{X}$ and $\mathcal{Y}$ where $|\mathcal{X}|$ is finite and $|\mathcal{Y}|$ is finite or countably infinite with $I(X,Y)< \infty$, there exists a RV $U$ supported on $\mathcal{U}$ such that $X$ and $U$ are independent, i.e., we have
 	\begin{align*}
 	I(U;X)=0,
 	\end{align*}
 	$Y$ is a deterministic function of $(U,X)$, i.e., we have 
 	\begin{align*}
 	H(Y|U,X)=0,
 	\end{align*}
 	$I(X;U|Y)$ can be upper bounded as follows
 	\begin{align*}
 	I(X;U|Y)\leq \log(I(X;Y)+1)+4,
 	\end{align*}
 	and 
 	$
 	|\mathcal{U}|\leq |\mathcal{X}|(|\mathcal{Y}|-1)+2.
 	$
 \end{lemma}
 In the next lemma we introduce $\bar{U}=(\bar{U}_1,..,\bar{U}_N)$ and provide some properties. 
 \begin{lemma}\label{pare}
 	 Let $\bar{U}=(\bar{U}_1,..,\bar{U}_N)$ where for all $i\in\{1,...,N\}$, $\bar{U}_i$ is a discrete RV defined on the alphabet $\bar{\mathcal{U}}_i=\mathcal{U}\times \mathcal{X}_1\times...\times \mathcal{X}_{i-1} $. 
 	For any feasible $U$ in \eqref{main1}, let $(\bar{U},X,Y,U)$ have the following joint distribution
 	\begin{align}\label{dorost}
 	P_{\bar{U},X,Y,U}(\bar{u},x,y,u)=P_{X,Y,U}(x,y,u)\prod_{i=1}^{N} P_{\bar{U}_i|X_i}(\bar{u}_i|x_i).
 	\end{align}
 	Furthermore, for all $i\in\{1,...,N\}$ let 
 	\begin{align}\label{dorost2}
 	P_{\bar{U}_i|X_i}(\bar{u}_i|x_i)=P_{X_1,...,X_{i-1},U|X_i}(x_1^{(i)},...,x_{i-1}^{(i)},u^{(i)}|x_i),
 	\end{align}
 	where $u_i=(x_1^{(i)},...,x_{i-1}^{(i)},u^{(i)})$. Then, we have
 	\begin{itemize}
 		\item [i.] $\bar{U}-X-(Y,U)$ forms a Markov chain.
 		\item [ii.] $\{(\bar{U}_i,Y_i,X_i)\}_{i=1}^N$ are mutually independent.
 	\end{itemize}
 \end{lemma}
 \begin{proof}
 	The proof of (i) is similar as proof of \cite[Th.~1]{9457633}. Furthermore, for proving (ii) note that by assumption we have $P_{XY}(x,y)=\prod_{i=1}^N P_{X_iY_i}(x_i,y_i)$ so that the same proof as \cite[Th.~1]{9457633} works. 
 \end{proof}
 In the next lemma we show that for any any feasible $U$ in \eqref{main1}, $\bar{U}$ has the same privacy leakage as $U$.
 \begin{lemma}\label{pare2}
 	for any any feasible $U$ in \eqref{main1}, let $\bar{U}$ be the RV that is constructed in Lemma~\ref{pare}. We have
 	\begin{align}\label{dorost3}
 	I(X;U)=I(X;\bar{U}).
 	\end{align}
 \end{lemma}
 \begin{proof}
 	The proof is based on \eqref{dorost2} and the fact that $H(X_i|\bar{U}_i)=H(X_i|X_{i-1},..,X_1,\bar{U})$. By checking the proof of privacy leakage as stated in \cite[Th.~1]{9457633}, the assumption that $X$ is an element-wise deterministic function of $Y$ has not been used. Thus, the same proof can be used here.
 \end{proof}
 Next theorem is similar to \cite[Th.~1]{9457633} which helps us to derive upper bounds on privacy-utility trade-off defined in \eqref{main1}.
 \begin{theorem}\label{theorem1}
 		For any feasible $U$ in \eqref{main1}, there exists RV $U^*=(U^*_1,...,U^*_N)$ with conditional distribution $P_{U^*|Y}(u^*|y)=\prod_{i=1}^NP_{U^*_i|Y_i}(u^*_i|y_i)$ that obtain the same leakage as $U$,i.e, we have
 		\begin{align}\label{koon}
 		I(X;U^*)=I(X;U),
 		\end{align} 
 		and for all $i\in\{1,...,K\}$ it bounds the utility achieved by $U$ as follows
 		\begin{align}\label{koonkesh}
 		I(C_i;U)\leq I(C_i;U^*)+\min\{\Delta_i^1,\Delta_i^2\}, 
 		\end{align}
 		where
 		\begin{align*}
 		\Delta_i^1&=I(X;C_i)+\sum_{j:Y_i\in C_j}H(X_j|Y_j)\\&=\sum_{j:Y_i\in C_j} \left( I(X_j;Y_j)+H(X_j|Y_j)\right),\\
 		\Delta_i^2&=I(X;C_i)+\sum_{j:Y_i\in C_j}\left(\log(I(X_j;Y_j)+1)+4\right)\\&=\sum_{j:Y_i\in C_j} \left( I(X_j;Y_j)+\log(I(X_j;Y_j)+1)+4\right).
 		\end{align*}
 \end{theorem}
\begin{proof}
	The proof is provided in Appendix~A.
\end{proof}
\begin{remark}
	Theorem~\ref{theorem1} is an extension of \cite[Th.~1]{9457633} for correlated $X_i$ and $Y_i$, $\forall i\in\{1,...,N\}$. 
\end{remark}
Let $\delta_i=\min\{\Delta_i^1,\Delta_i^2\}$. By using Theorem~\ref{theorem1} we can decompose $I(U;Y)$ and $I(U;X)$ into $N$ parts and derive an upper bound on \eqref{main1} as follows. We show that  without loss of optimality we can replace the constraint $\sum_{i=1}^NI(X_i;U_i)\leq \epsilon$ by $N$ individual constraints using auxiliary variables $\{\epsilon_i\}_{i=1}^N$. 
\begin{lemma}\label{cool}
	\begin{align}
	h_{\epsilon}^K(P_{XY})&\leq \!\!\!\!\!\!\!\sup_{\begin{array}{c}\substack{\{P_{U_i|X_i,Y_i}\}_{i=1}^N:\\\sum_{i=1}^NI(X_i;U_i)\leq \epsilon}\end{array}}\!\!\sum_{i=1}^N\!\left(\sum_{j:Y_i\in C_j}\!\!\!\lambda_j\right)\!\!\left(I(U_i;Y_i)\!+\!\delta_i\right)\label{ame1}\\&=\!\!\!\!\!\!\!\sup_{\begin{array}{c}\substack{\{P_{U_i|X_i,Y_i},\\\epsilon_i\}_{i=1}^N:\\I(X_i;U_i)\leq\epsilon_i,\\ 0\leq\epsilon_i\leq \epsilon,\\ \sum_i \epsilon_i\leq \epsilon}\end{array}}\!\!\sum_{i=1}^N\!\left(\sum_{j:Y_i\in C_j}\!\!\!\lambda_j\right)\!\!\left(I(U_i;Y_i)\!+\!\delta_i\right)\label{nane1}.
	\end{align}
\end{lemma}
\begin{proof}
	For any feasible $\{U_i\}_{i=1}^N$ satisfying the constraints in \eqref{nane1}, we have $\sum_{i=1}^NI(X_i;U_i)\leq \epsilon$. Hence, $\eqref{nane1}\leq\eqref{ame1}$. Now let $\{\bar{U}_i\}_{i=1}^N$ be the maximizer of \eqref{ame1}. In this case we can replace the constraint $\sum_{i=1}^NI(X_i;U_i)\leq \epsilon$ by $I(X;U_i)\leq I(X;\bar{U}_i), \forall i$. We have
	\begin{align}
	&\sup_{\begin{array}{c}\substack{\{P_{U_i|X_i,Y_i}\}_{i=1}^N:\\\sum_{i=1}^NI(X_i;U_i)\leq \epsilon}\end{array}}\!\!\sum_{i=1}^N\!\left(\sum_{j:Y_i\in C_j}\!\!\!\lambda_j\right)\!\!\left(I(U_i;Y_i)\!+\!\delta_i\right)\nonumber\\&\stackrel{(a)}{=}\!\!\!\!\!\!\!\sup_{\begin{array}{c}\substack{\{P_{U_i|X_i,Y_i}\}_{i=1}^N:\\I(X_i;U_i)\leq I(X_i;\bar{U}_i),\\ \sum_i I(X_i;\bar{U}_i)\leq \epsilon}\end{array}}\!\!\sum_{i=1}^N\!\left(\sum_{j:Y_i\in C_j}\!\!\!\lambda_j\right)\!\!\left(I(U_i;Y_i)\!+\!\delta_i\right)\label{koonmoon}\\&\leq \!\!\!\!\!\!\!\sup_{\begin{array}{c}\substack{\{P_{U_i|X_i,Y_i},\\\epsilon_i\}_{i=1}^N:\\I(X_i;U_i)\leq\epsilon_i,\\ 0\leq\epsilon_i\leq \epsilon,\\ \sum_i \epsilon_i\leq \epsilon}\end{array}}\!\!\sum_{i=1}^N\!\left(\sum_{j:Y_i\in C_j}\!\!\!\lambda_j\right)\!\!\left(I(U_i;Y_i)\!+\!\delta_i\right)\nonumber.
	\end{align}
	Step (a) holds since the maximizer of \eqref{ame1} satisfies the constraints in \eqref{koonmoon} and the fact that the constraints in \eqref{koonmoon} also results in $\sum_{i=1}^NI(X_i;U_i)\leq \epsilon$. The last inequality states that $\eqref{ame1}\leq\eqref{nane1}$.
	Thus, $\eqref{nane1}=\eqref{ame1}$.
\end{proof}
 Note that we can maximize \eqref{nane1} in two phases. First, we fix $\epsilon_i$ and decompose \eqref{nane1} into $N$ problems and find an upper bound for each of them. We then show that the upper bound can be written as a linear program over $\{\epsilon_i\}_{i=1}^N$ and we derive the final upper bound.
  For fixed $\epsilon_i$, $i\in\{1,...,N\}$, \eqref{nane1} can be decomposed into $N$ parallel privacy problems, where each privacy problem can be stated as follows
 \begin{align}\label{sd}
 h_{\epsilon}^i(P_{X_iY_i})=\sup_{\begin{array}{c} 
 	\substack{P_{U|Y,X}: I(U_i;X_i)\leq\epsilon_i,}
 	\end{array}}I(Y_i;U_i).
 \end{align}
 The last argument holds since the terms $\sum_{j:Y_i\in C_j}\lambda_j$ and $\sum_{i=1}^N\sum_{j:Y_i\in C_j}\lambda_j\delta_i$ are constant. 
 The problem \eqref{sd} has been studied in \cite{zamani2022bounds} and \cite{kostala}, where an upper bound has been derived. Using \cite[Lemma~8]{zamani2022bounds} we have
 \begin{align}\label{sm}
 h_{\epsilon}^i(P_{X_iY_i})\leq H(Y_i|X_i)+\epsilon_i, \ \forall i\in\{1,...,N\}.
 \end{align} 
 As shown in \cite[Corollary~2]{zamani2022bounds} the upper bound in \eqref{sm} is tight if $X_i$ is a deterministic function of $Y_i$. In the next result, by using Lemma~\ref{cool} and \eqref{sm} we obtain the final upper bound. 
 \begin{theorem}\label{theorem2}
 	For any $(X,Y)$ distributed according $P_{XY}(x,y)=\prod_{i=1}^N P_{X_i,Y_i}(x_i,y_i)$ and any $\epsilon < I(X;Y)$ we have
 	\begin{align}\label{jaleb}
 	&h_{\epsilon}^K(P_{XY})\!\leq\nonumber\\ &\epsilon\max_i\left\{\sum_{j:Y_i\in C_j}\!\!\!\lambda_j\right\}\!+\!\sum_{i=1}^N\!\left(\sum_{j:Y_i\in C_j}\!\!\!\lambda_j\right)\!\!\left(H(Y_i|X_i)\!+\!\delta_i\right),
 	\end{align}
 	where $\delta_i=\min\{\Delta_i^1,\Delta_i^2\}$, $\Delta_i^1$ and $\Delta_i^2$ are defined in Theorem~\ref{theorem1}.
\end{theorem}
 \begin{proof}
 	The proof directly follows from Theorem~\ref{theorem1}, Lemma~\ref{cool} and \eqref{sm}. We have
 	\begin{align*}
 	&h_{\epsilon}^K(P_{XY})\!\leq\!\!\!\!\!\!\! \sup_{\begin{array}{c}\substack{\{\epsilon_i\}_{i=1}^N:\\ 0\leq\epsilon_i\leq \epsilon,\\ \sum_i \epsilon_i\leq \epsilon}\end{array}}\!\!\!\sum_{i=1}^N\!\left(\sum_{j:Y_i\in C_j}\!\!\!\lambda_j\right)\!\!\left(H(Y_i|X_i)\!+\!\epsilon_i\!+\!\delta_i\right)\\&=\epsilon\max_i\left\{\sum_{j:Y_i\in C_j}\!\!\!\lambda_j\right\}\!+\!\sum_{i=1}^N\!\left(\sum_{j:Y_i\in C_j}\!\!\!\lambda_j\right)\!\!\left(H(Y_i|X_i)\!+\!\delta_i\right).
 	\end{align*}
 	The last line follows since the terms $\sum_{j:Y_i\in C_j}\lambda_j$, $\sum_{i=1}^N\sum_{j:Y_i\in C_j}\lambda_j\delta_i$ and $\sum_{i=1}^N\sum_{j:Y_i\in C_j}\lambda_jH(X_i|Y_i)$ are constant, thus, we choose $\epsilon_i=\epsilon$ with the largest coefficient. 
 \end{proof}
 \subsection{Lower bounds on $\bm{h_{\epsilon}^K(P_{XY})}$:}
 To find the first lower bound, let us construct $\tilde{U}=(\tilde{U}_1,..,\tilde{U}_N)$ as follows: 
 For all $i\in\{1,...,N\}$, let $\tilde{U}_i$ be the RV found by the EFRL in \cite[Lemma~3]{zamani2022bounds} using $X\leftarrow X_i$, $Y\leftarrow Y_i$ and $\epsilon\leftarrow \epsilon_i$. Thus, we have
 \begin{align}\label{kk}
 I(\tilde{U}_i;X_i)=\epsilon_i,
 \end{align}
 and
 \begin{align}\label{mm}
 H(Y_i|X_i,\tilde{U}_i)=0,
 \end{align}
 where $0\leq\epsilon_i\leq \epsilon$. 
 The utility attained by $\tilde{U}_i$ can be lower bounded as follows
 \begin{align}\label{ja}
 &I(\tilde{U}_i;Y_i)\nonumber\\
 &=I(\tilde{U}_i;X_i)+H(Y_i|X_i)-I(X;\tilde{U}_i|Y_i)-H(Y_i|\tilde{U}_i,X_i)\nonumber\\
 &\geq H(Y_i|X_i)-H(X_i|Y_i)+\epsilon_i, 
 \end{align}
 where in the last line we used \eqref{kk}, \eqref{mm} and $I(X;\tilde{U}_i|Y_i)\leq H(X_i|Y_i)$.
 Furthermore, similar to the upper bounds, we can construct $\tilde{U}$ so that $\{(\tilde{U}_i,Y_i,X_i)\}_{i=1}^N$ become mutually independent. Thus, we can find a first lower bound for $h_{\epsilon}^K(P_{XY})$ as in the next lemma.
 \begin{lemma}\label{lem4}
 	For any $(X,Y)$ distributed according $P_{XY}(x,y)=\prod_{i=1}^N P_{X_iY_i}(x_i,y_i)$ and any $\epsilon < I(X;Y)$ we have
 	\begin{align}\label{jaleb1}
 	&h_{\epsilon}^K(P_{XY})\!\geq\nonumber\\ &\epsilon\max_i\!\left\{\!\sum_{j:Y_i\in C_j}\!\!\!\!\lambda_j\!\right\}\!+\!\sum_{i=1}^N\!\left(\sum_{j:Y_i\in C_j}\!\!\!\!\lambda_j\!\!\right)\!\!\left(H(Y_i|X_i)\!-\!H(X_i|Y_i)\!\right).
 	\end{align}
 \end{lemma}
 \begin{proof}
 	The proof follows from \eqref{ja} and the fact that $\{(\tilde{U}_i,Y_i,X_i)\}_{i=1}^N$ are independent over different $i$'s. We choose $\epsilon_i=\epsilon$ with the largest coefficient.
 \end{proof}
\begin{remark}\label{isha}
	The privacy mechanism design that attains the lower bound in Lemma~\ref{lem4} asserts that we only release information about $X_i$ where $i=\arg \max_m \left\{\sum_{j:Y_m\in C_j}\lambda_j\right\}$ and the leakage from $X_i$ is equal to $\epsilon$ (maximum allowed leakage). In other words, the disclosed data that achieves \eqref{jaleb1} satisfies $I(U;X_i)=\epsilon$ and $I(U;X_j)=0, \ \forall j\neq i$, where $i$ is defined earlier. 
\end{remark}
 To find the second lower bound, let us construct $U'=(U'_1,..,U'_N)$ as follows:
 For all $i\in\{1,...,N\}$, let $U'_i$ be the RV found by the ESFRL in \cite[Lemma~4]{zamani2022bounds} using $X\leftarrow X_i$, $Y\leftarrow Y_i$ and $\epsilon\leftarrow \epsilon_i$. Thus, we have
 \begin{align}
 I(U'_i;X_i)&=\epsilon_i,\label{kkk}\\
 H(Y_i|X_i,U'_i)&=0,\label{kkkk}
 \end{align}
 and 
 \begin{align}
 I(X_i;U'_i|Y_i)&\leq \alpha_iH(X_i|Y_i)\nonumber\\&+(1-\alpha_i)\left(\log(I(X_i;Y_i)+1)+4\right),\label{kkkkk}
 \end{align}
 where $\alpha_i=\frac{\epsilon_i}{H(X_i)}$.
The utility attained by $U'_i$ can be lower bounded as follows
\begin{align}\label{jaa}
&I(U'_i;Y_i)=\nonumber\\ 
&I(\tilde{U}_i;X_i)+H(Y_i|X_i)-I(X;\tilde{U}_i|Y_i)-H(Y_i|\tilde{U}_i,X_i)\nonumber
\\&\geq H(Y_i|X_i)-\alpha_iH(X_i|Y_i)+\epsilon_i\nonumber\\&-(1-\alpha_i)\left(\log(I(X_i;Y_i)+1)+4\right), 
\end{align}
where in the last line we used \eqref{kkk}, \eqref{kkkk} and \eqref{kkkkk}.
 Furthermore, similar to the upper bounds, we can construct $U'$ so that $\{(U'_i,Y_i,X_i)\}_{i=1}^N$ become mutually independent. Hence, we can find the second lower bound for $h_{\epsilon}^K(P_{XY})$ as follows.
 \begin{lemma}\label{lem5}
 	Let $\beta_i=H(Y_i|X_i)-\alpha_iH(X_i|Y_i)+\epsilon_i-(1-\alpha_i)\left(\log(I(X_i;Y_i)+1)+4\right)$ and $\alpha_i=\frac{\epsilon_i}{H(X_i)}$. For any $(X,Y)$ distributed according $P_{XY}(x,y)=\prod_{i=1}^N P_{X_iY_i}(x_i,y_i)$ and any $\epsilon < I(X;Y)$ we have
 	\begin{align}\label{jaleb2}
 	&h_{\epsilon}^K(P_{XY})\!\geq \sup_{\begin{array}{c}\substack{\{\epsilon_i\}_{i=1}^N:\\ 0\leq\epsilon_i\leq \epsilon,\\ \sum_i \epsilon_i\leq \epsilon}\end{array}}\sum_{i=1}^N\left(\sum_{j:Y_i\in C_j}\!\!\!\lambda_j\right)\beta_i.
 	\end{align}
 \end{lemma}
 \begin{proof}
 	The proof follows from \eqref{jaa} and the fact that $\{(\tilde{U}_i,Y_i,X_i)\}_{i=1}^N$ are independent over different $i$'s. 
 \end{proof}
 Note that the right hand side of \eqref{jaleb2} is a linear program which can be rewritten as follows: Let $\gamma_i=1-\frac{H(X_i|Y_i)}{H(X_i)}+\frac{\log(I(X_i;Y_i)+1)+4}{H(X_i)}$ and $\mu_i=\sum_{j:Y_i\in C_j}\lambda_j$. We have
 \begin{align*}
 &\sum_{i=1}^N\left(\sum_{j:Y_i\in C_j}\!\!\!\lambda_j\right)\beta_i\!\\&=\!\sum_{i=1}^N \mu_i\left(H(Y_i|X_i)\!-\!(\log(I(X_i;Y_i)\!+\!1)\!+\!4)\right)\!+\!\mu_i\gamma_i\epsilon_i
 \end{align*} 
 Thus,
 \begin{align}\label{tala}
 \sup_{\begin{array}{c}\substack{\{\epsilon_i\}_{i=1}^N:\\ 0\leq\epsilon_i\leq \epsilon,\nonumber\\ \sum_i \epsilon_i\leq \epsilon}\end{array}}&\sum_{i=1}^N\left(\sum_{j:Y_i\in C_j}\!\!\!\lambda_j\right)\beta_i\\=\sum_{i=1}^N &\mu_i\left(H(Y_i|X_i)\!-\!(\log(I(X_i;Y_i)\!+\!1)\!+\!4)\right)\nonumber\\+&\max_i\{\mu_i\gamma_i\}\epsilon.
 \end{align}
 Where in the last line we choose $\epsilon_i=\epsilon$ with the largest coefficient.
 \begin{theorem}\label{theorem3}
 	Let $\gamma_i=1-\frac{H(X_i|Y_i)}{H(X_i)}+\frac{\log(I(X_i;Y_i)+1)+4}{H(X_i)}$ and $\mu_i=\sum_{j:Y_i\in C_j}\lambda_j$. For any $(X,Y)$ distributed according $P_{XY}(x,y)=\prod_{i=1}^N P_{X_i,Y_i}(x_i,y_i)$ and any $\epsilon < I(X;Y)$ we have
 	\begin{align*}
 	h_{\epsilon}^K(P_{XY})\geq \max\{L_{\epsilon}^1,L_{\epsilon}^2\},
 	\end{align*}
 	where
 	\begin{align*}
 	&L_{\epsilon}^1=\epsilon\max_i\{\mu_i\}+\sum_{i=1}^N\mu_i\left(H(Y_i|X_i)-H(X_i|Y_i)\right),\\
 	&L_{\epsilon}^2=\\&\sum_{i=1}^N\! \mu_i\left(H(Y_i|X_i)\!-\!(\log(I(X_i;\!Y_i)\!+\!1)\!+\!4)\right)\nonumber\!+\!\max_i\{\mu_i\gamma_i\}\epsilon.
 	\end{align*}
 \end{theorem}
 \begin{proof}
 	The proof directly follows from Lemma~\ref{lem4} and \eqref{tala}.
 \end{proof}
\begin{remark}\label{isha2}
	Similar to Remark~\ref{isha}, the disclosed data that attains the lower bound \eqref{tala} only releases the information of $X_i$ with leakage equal to $\epsilon$, where $i=\arg \max_m \left\{\mu_m\gamma_m\right\}=\arg \max_m \left\{\left(\sum_{j:Y_i\in C_j}\lambda_j\right)\gamma_m\right\}$ and $\gamma_m=1-\frac{H(X_m|Y_i)}{H(X_m)}+\frac{\log(I(X_m;Y_m)+1)+4}{H(X_m)}$. In this case, we have $I(U;X_i)=\epsilon$ and $U$ is independent of $X_j$ for $j\neq i$. 
\end{remark}
\begin{corollary}\label{tight}
	The upper bound in \eqref{jaleb} is tight within at most $\sum_{i=1}^N \left(\sum_{j:Y_i\in C_j}\lambda_j\right)(\delta_i+H(X_i|Y_i))$ nats, since by using Theorem~\ref{theorem2} and Theorem~\ref{theorem3} the distance between the lower bound $L_{\epsilon}^1$ and the upper bound in \eqref{jaleb} is $\sum_{i=1}^N \left(\sum_{j:Y_i\in C_j}\lambda_j\right)(\delta_i+H(X_i|Y_i))$ nats.
\end{corollary}
\subsection{Special case $\bm\epsilon=0$ (Independent $\bm X$ and $\bm U$):}
In this section we derive lower and upper bounds for $h_{0}^K(P_{XY})$. When $\epsilon=0$, \eqref{nane1} is decomposed into $N$ parallel problems as follows.
\begin{align}
h_{0}^i(P_{X_i,Y_i}) =\mu_i\delta_i+\mu_i \sup_{\begin{array}{c}\substack{P_{U_i|X_i,Y_i}:\\I(X_i;U_i)=0}\end{array}}\!\!I(U_i;Y_i)
\end{align}
In this case we obtain the next result.
\begin{corollary}\label{ch}
	In case of perfect privacy, i.e., zero leakage, let $\epsilon=0$ in Theorem~\ref{theorem2} and Theorem~\ref{theorem3}, then we have
	\begin{align}\label{jaleb0}
	\max\{L_{0}^1,L_{0}^2\}\leq h_{0}^K(P_{XY})\!\leq\sum_{i=1}^N\!\mu_i\left(H(Y_i|X_i)\!+\!\delta_i\right),
	\end{align}
	where
	\begin{align*}
	&L_{0}^1=\sum_{i=1}^N\mu_i\left(H(Y_i|X_i)-H(X_i|Y_i)\right),\\
	&L_{0}^2=\sum_{i=1}^N\! \mu_i\left(H(Y_i|X_i)\!-\!(\log(I(X_i;\!Y_i)\!+\!1)\!+\!4)\right),
	\end{align*}
	$\delta_i$ is defined in Theorem~\ref{theorem2} and $\mu_i$ is defined in Theorem~\ref{theorem3}. 
\end{corollary}
A new upper bound on $h_{0}^K(P_{XY})$ can be derived by strengthening \eqref{sd} with zero leakage. Using \cite[Theorem 4]{zamani2022bounds} we have
\begin{align}\label{jen}
h_{0}^i(P_{X_i,Y_i})\leq \min\{U^{i_1}_0,U^{i_2}_0\}
\end{align}
where
	\begin{align*}
&U^{i_1}_0= H(Y_i|X_i),\\
&U^{i_2}_0 =  H(Y_i|X_i) +\sum_{y_i\in\mathcal{Y}_i}\int_{0}^{1} \mathbb{P}_{X_i}\{P_{Y_i|X_i}(y_i|X_i)\geq t\}\times\\&\log (\mathbb{P}_{X_i}\{P_{Y_i|X_i}(y_i|X_i)\geq t\})dt+I(X_i;Y_i).
\end{align*}
Furthermore, if $|\mathcal{Y}_i|=2$, then $U^{i_2}_0$ is tight.\\
Combining \eqref{nane1} and \eqref{jen} we have
\begin{align*}
 h_{0}^K(P_{XY})\!\leq\sum_{i=1}^N\!\mu_i\left(\min\{U^{i_1}_0,U^{i_2}_0\}\!+\!\delta_i\right).
\end{align*}
\subsection{Special case $\bm{ X_i=f_i(Y_i) \ (H(X_i|Y_i)=0)}$:}
For all $i\in\{1,..,N\}$, let $X_i$ be a deterministic function of $Y_i$, i.e., $X_i=f_i(Y_i)$. Using Corollary~\ref{ch} the gap between the upper bound and lower bound for $h_{\epsilon}^K(P_{XY})$ is lower than or equal to $\sum_i \mu_i\delta_i$, where $\delta_i=\Delta_i^1=\sum_{j:Y_i\in C_j}H(X_j)$. However, in this case we can modify the proof of Theorem~\ref{theorem1} as has been done in \cite[Theorem 1]{9457633} which results in $\delta_i=0$. Thus, by using \cite[Theorem 1]{9457633} and combining it with \eqref{nane1}, the gap between upper and lower bound on $h_{\epsilon}^K(P_{XY})$ become zero and we have
\begin{align}
&h_{\epsilon}^K(P_{XY})\!=\nonumber\\ &\epsilon\max_i\!\left\{\!\sum_{j:Y_i\in C_j}\!\!\!\!\lambda_j\!\right\}\!+\!\sum_{i=1}^N\!\left(\sum_{j:Y_i\in C_j}\!\!\!\!\lambda_j\!\!\right)\!\!H(Y_i|X_i).
\end{align}
In this case the lower bound attained by Lemma~\ref{lem4} is larger than or equal to the lower bound attained by \eqref{tala}, i.e., $L_{\epsilon}^1\geq L_{\epsilon}^2$, since we have
\begin{align*}
&H(Y_i|X_i)+\epsilon_i\geq \\&H(Y_i|X_i)+\epsilon_i-(1-\alpha_i)(\log(H(X_i)+1)+4).
\end{align*}  
\section{Conclusion}
An information theoretic privacy design problem with $K$ users has been studied, where it has been shown that by using a specific transformation, Functional Representation Lemma and Strong Functional Representation Lemma upper bounds on $h_{\epsilon}^K(P_{XY})$ can be derived and the upper bounds can be decomposed into $N$ parallel privacy problems. Furthermore, lower bounds are obtained using Functional Representation Lemma and Strong Functional Representation Lemma. The gap between the upper and lower bound is lower than or equal to $\sum_{i=1}^N \left(\sum_{j:Y_i\in C_j}\lambda_j\right)(\delta_i+H(X_i|Y_i))$ which can be reduced to zero when $X$ is a point-wise deterministic function of $Y$. The lower bound found in this work asserts that the disclosed data releases the information about $X_i$ with leakage equal to $\epsilon$ (the maximum possible leakage) and it is independent of all other components of $X$. Finally, in perfect privacy case where no leakages are allowed, i.e., $\epsilon=0$, the upper bound is strengthened. 
 \section*{Appendix A}
 	For any any feasible $U$ in \eqref{main1}, let $\bar{U}=(\bar{U}_1,..,\bar{U}_N)$ be the RV constructed as in Lemma~~\ref{pare} and $\tilde{U}=(\tilde{U}_1,..,\tilde{U}_N)$ be constructed as follows: For all $i\in\{1,...,N\}$, let $\tilde{U}_i$ be the RV found by the FRL (Lemma~\ref{lemma1}) using $X\leftarrow(\bar{U}_i,X_i)$ and $Y\leftarrow Y_i$. Thus, we have
 \begin{align}\label{jj}
 I(\tilde{U}_i;\bar{U}_i,X_i)&=0,
 \end{align}
 and
 \begin{align}\label{ii}
 H(Y_i|X_i,\bar{U}_i,\tilde{U}_i)&=0.
 \end{align}
 Let $U^*=(U^*_1,...,U^*_N)$ where $U^*_i=(\tilde{U}_i,\bar{U}_i)$. Using Lemma~\ref{pare} part (ii), i.e., $\{(\bar{U}_i,X_i,Y_i)\}_{i=1}^N$ are mutually independent, and the construction used as in proof of \cite[Lemma~1]{kostala} we can build $\tilde{U}_i$ such that $\{(\tilde{U}_i,\bar{U}_i,X_i,Y_i)\}_{i=1}^N$ are mutually independent. Due to the independence of $(\tilde{U}_i,\bar{U}_i,Y_i)$ the conditional distribution $P_{U^*|Y}(u^*|y)$ can be stated as follows
 \begin{align}
 P_{U^*|X}(u^*|x)\!=\!\prod_{i=1}^{N}P_{\tilde{U}_i,\bar{U}_i|Y}(\tilde{U}_i,\bar{U}_i|y)\!=\!\prod_{i=1}^{N}P_{U^*|Y}(u^*_i|y_i).
 \end{align}
 For proving leakage constraint in \eqref{koon} we have
 \begin{align*}
 I(X;U^*)&= I(X_1,..,X_N;\tilde{U}_1,\bar{U}_1,...,\tilde{U}_N,\bar{U}_N)\\&\stackrel{(a)}{=} \sum_i I(X_i;\tilde{U}_i,\bar{U}_i)\\&\stackrel{(b)}{=}\sum_i I(X_i;\bar{U}_i) = I(X;\bar{U}) \\&\stackrel{(c)}{=} I(X;U),  
 \end{align*}
 where (a) follows by the independency of $(\tilde{U}_i,\bar{U}_i,Y_i)$ over different $i$'s, (b) follows by \eqref{jj} and (c) holds due to Lemma~\ref{pare2}. 
 For proving the upper bound on utility in \eqref{koonkesh} we first derive expressions for $I(C_i;U^*)$ and $I(C_i;U)$ as follows.
 \begin{align}
 I(C_i;U^*)&=I(X,C_i;U^*)-I(X;U^*|C_i)\nonumber\\
 &=I(X;U^*)+I(C_i;U^*|X)-I(X;U^*|C_i)\nonumber\\
 &=I(X;U^*)+H(C_i|X)\nonumber\\& \ \ \ -H(C_i|X,U^*)-I(X;U^*|C_i).\label{kosenanat}
 \end{align} 
 Similarly, we have
 \begin{align}
 I(C_i;U)=\!I(X;U)\!\!+\!\!H(C_i|X)\!\!-\!\!H(C_i|X,U)\!\!-\!\!I(X;U|C_i\!).\label{kosebabat}
 \end{align}
 Thus, by using \eqref{kosenanat}, \eqref{kosebabat} and \eqref{koon} we have
 \begin{align*}
 I(C_i;U)&\stackrel{(a)}{=}I(C_i;U^*)+H(C_i|X,U^*)\\& \ \ +I(X;U^*|C_i)-H(C_i|X,U)-I(X;U|C_i)\\
 &\stackrel{(b)}{=}I(C_i;U^*)+\sum_{j:Y_i\in C_j} H(Y_i|X_i,U^*_i)\\& \ \ +\sum_{j:Y_i\in C_j} I(X_j;U^*_j|Y_j)+\sum_{j:Y_i\notin C_j} I(Xj;U^*_j)\\ \ \ &-H(C_i|X,U)-I(X;U|C_i)\\&\leq I(C_i;U^*)+\sum_{j:Y_i\in C_j} H(Y_j|X_j,U^*_j)\\& \ \ +\sum_{j:Y_i\in C_j} I(X_j;U^*_j|Y_j)+\sum_{j:Y_i\notin C_j} I(X_j;U^*_j)\\& \ \ -I(X;U|C_i) \\&= I(C_i;U^*)+\sum_{j:Y_i\in C_j} H(Y_j|X_j,\bar{U}_j,\tilde{U}_j)\\& \ \ +\!\!\sum_{j:Y_i\in C_j}\!\! I(X_j;\bar{U}_j,\tilde{U}_j|Y_j)\!+\!\!\!\!\sum_{j:Y_i\notin C_j}\!\! I(X_j;\bar{U}_j,\tilde{U}_j)\\& \ \ -I(X;U|C_i)\\&\stackrel{(c)}{=}I(C_i;U^*)+\!\!\sum_{j:Y_i\in C_j}\!\! I(X_j;\bar{U}_j,\tilde{U}_j|Y_j)\!\\& \ \ +\!\!\!\!\sum_{j:Y_i\notin C_j}\!\! I(X_j;\bar{U}_j)-I(X;U|C_i)\\
 &\stackrel{(d)}{\leq} I(C_i;U^*) +\!\!\sum_{j:Y_i\in C_j}\!\! H(X_j|Y_j)\!+ I(X;U)\\& \ \ -I(X;U|C_i)
 \\&=I(C_i;U^*)\!+\!\!\sum_{j:Y_i\in C_j}\!\! H(X_j|Y_j)\!+I(X;C_i)\\& \ \ -I(X;C_i|U)
 \\&\leq I(C_i;U^*)+\!\!\!\sum_{j:Y_i\in C_j}\!\! H(X_j|Y_j)\!+I(X;C_i)\\&=I(C_i;U^*)+\Delta_i^1.
 \end{align*}
 where in step (a) we used the equality property of the leakage in \eqref{koon},  where step (b) follows by the independency of $(\tilde{U}_i,\bar{U}_i,Y_i)$ over different $i$'s and step (c) follows by the fact that $\tilde{U}_i$ is produced by FRL, i.e.,  $H(Y_i|X_i,\bar{U}_i,\tilde{U}_i)=0$ and $\tilde{U}_i$ is independent of $\bar{U}_i$ and $X_i$. Moreover, in step (d) we used the simple bound $\sum_{j:Y_i\notin C_j} I(X_j;\bar{U}_j,\tilde{U}_j|Y_j)\leq \sum_{j:Y_i\notin C_j}H(X_j|Y_j)$ and $\sum_{j:Y_i\notin C_j}\!\! I(X_j;\bar{U}_j)\leq I(X;\bar{U})=I(X;U)$, where the last equality follows by Lemma~\ref{pare2}.\\
 For proving the second upper bound let $U'=(U'_1,..,U'_N)$ be constructed as follows. For all $i\in\{1,...,N\}$, let $U'_i$ be the RV found by SFRL in Lemma~\ref{lemma2} using $X\leftarrow(\bar{U}_i,X_i)$ and $Y\leftarrow Y_i$. Thus, we have
 \begin{align}\label{jjj}
 I(U'_i;\bar{U}_i,X_i)&=0,
 \end{align}
 and
 \begin{align}\label{iii}
 H(Y_i|X_i,\bar{U}_i,U'_i)&=0,
 \end{align}
 and
 \begin{align}\label{k}
 I(\bar{U}_i,X_i;U'_i|Y_i)&\leq \log(I(\bar{U}_i,X_i;Y_i)+1)+4\nonumber\\&=\log(I(X_i;Y_i)+1)+4,
 \end{align}
 where in last line we used the Markov chain stated in Lemma~\ref{pare2} part (i). Now let $U^*=(U^*_1,...,U^*_N)$ where $U^*_i=(U'_i,\bar{U}_i)$. Similarly, we can construct $U'$ such that $\{(U'_i,\bar{U}_i,X_i,Y_i)\}_{i=1}^N$ are mutually independent. The proof of \eqref{koon} does not change and for proving the second upper bound on utility we have
 \begin{align*}
 I(C_i;U)&\leq I(C_i;U^*)+\!\!\sum_{j:Y_i\in C_j}\!\! I(X_j;\bar{U}_j,U'_j|Y_j)\!\\& \ \ +\!\!\!\!\sum_{j:Y_i\notin C_j}\!\! I(X_j;\bar{U}_j)-I(X;U|C_i)\\&= I(C_i;U^*)\!+\!\!\!\!\!\!\sum_{j:Y_i\in C_j}\!\! I(X_j;U'_j|Y_j,\bar{U}_j)\!\\& \ \ +\!\!\!\sum_{j:Y_i\in C_j}I(X_j;\bar{U}_j|Y_j)+\sum_{j:Y_i\notin C_j}\!\!\! I(X_j;\bar{U}_j)\!\\&-\!I(X;U|C_i)\\&\stackrel{(a)}{\leq}  I(C_i;U^*)+\sum_{j:Y_i\in C_j}\!\!\left(\log(I(X_i;Y_i)+1)+4\right)\\&\ \ +\!\!\!\sum_{j:Y_i\in C_j}I(X_j;\bar{U}_j|Y_j)+\sum_{j:Y_i\notin C_j}\!\!\! I(X_j;\bar{U}_j)\!\\&-\!I(X;U|C_i)\\&\stackrel{(b)}{=}I(C_i;U^*)+\sum_{j:Y_i\in C_j}\!\!\left(\log(I(X_i;Y_i)+1)+4\right)\\&I(X;C_i)-I(X;C_i|U)-\sum_{j:Y_i\notin C_j}I(\bar{U}_j;Y_j)\\&\leq I(C_i;U^*)+\sum_{j:Y_i\in C_j}\!\!\left(\log(I(X_i;Y_i)+1)+4\right)\\&\ \ +I(X;C_i)=I(C_i;U^*)+\Delta_i^2.
 \end{align*}
 Step (a) follows by \eqref{k} since we have
 \begin{align*}
 I(X_i;U'_i|Y_i,\bar{U}_i)\leq I(\bar{U}_i,X_i;U'_i|Y_i).
 \end{align*}
 Furthermore, step (b) follows since we have
 \begin{align*}
 &\sum_{j:Y_i\in C_j}\!\!I(X_j;\bar{U}_j|Y_j)+\!\!\sum_{j:Y_i\notin C_j}\!\!\! I(X_j;\bar{U}_j)-I(X;U|C_i)\stackrel{(i)}{=}\\&\sum_{j:Y_i\in C_j}\!\! H(\bar{U}_j|Y_j)\!\!+\!\!\sum_{j:Y_i\notin C_j} \!\!H(\bar{U}_j)\!-\!H(\bar{U}|X)\!-\!I(X;U|C_i)=\\&\sum_{j:Y_i\in C_j}\!\! H(\bar{U}_j|Y_j)\!\!+\!\!\sum_{j:Y_i\notin C_j}\!\!\! \!\!H(\bar{U}_j)\!-\!H(\bar{U})\!\\&\ \ I(\bar{U};X)\!-\!I(X;U|C_i)\stackrel{(ii)}{=}\\& \ \ -\sum_{j:Y_i\in C_j} I(\bar{U}_j;Y_j)+I(X;U)-I(X;U|C_i)=\\&\ \ -\sum_{j:Y_i\in C_j} I(\bar{U}_j;Y_j)+I(X;C_i)-I(X;C_i|U),
 \end{align*}
 where step (i) follows by the Markov chain $\bar{U}-X-Y$ in Lemma~\ref{pare} and in step (ii) we used Lemma~\ref{pare2}.
\bibliographystyle{IEEEtran}
\bibliography{IEEEabrv,IZS}

\begin{thebibliography}{10}
\providecommand{\url}[1]{#1}
\csname url@samestyle\endcsname
\providecommand{\newblock}{\relax}
\providecommand{\bibinfo}[2]{#2}
\providecommand{\BIBentrySTDinterwordspacing}{\spaceskip=0pt\relax}
\providecommand{\BIBentryALTinterwordstretchfactor}{4}
\providecommand{\BIBentryALTinterwordspacing}{\spaceskip=\fontdimen2\font plus
\BIBentryALTinterwordstretchfactor\fontdimen3\font minus
  \fontdimen4\font\relax}
\providecommand{\BIBforeignlanguage}[2]{{%
\expandafter\ifx\csname l@#1\endcsname\relax
\typeout{** WARNING: IEEEtran.bst: No hyphenation pattern has been}%
\typeout{** loaded for the language `#1'. Using the pattern for}%
\typeout{** the default language instead.}%
\else
\language=\csname l@#1\endcsname
\fi
#2}}
\providecommand{\BIBdecl}{\relax}
\BIBdecl

\bibitem{Calmon2}
H.~{Wang}, L.~{Vo}, F.~P. {Calmon}, M.~{M\'{e}dard}, K.~R. {Duffy}, and
  M.~{Varia}, ``Privacy with estimation guarantees,'' \emph{IEEE Transactions
  on Information Theory}, vol.~65, no.~12, pp. 8025--8042, Dec 2019.

\bibitem{yamamoto}
H.~Yamamoto, ``A source coding problem for sources with additional outputs to
  keep secret from the receiver or wiretappers (corresp.),'' \emph{IEEE
  Transactions on Information Theory}, vol.~29, no.~6, pp. 918--923, 1983.

\bibitem{issa}
I.~{Issa}, S.~{Kamath}, and A.~B. {Wagner}, ``An operational measure of
  information leakage,'' in \emph{2016 Annual Conference on Information Science
  and Systems}, March 2016, pp. 234--239.

\bibitem{makhdoumi}
A.~Makhdoumi, S.~Salamatian, N.~Fawaz, and M.~M{\'e}dard, ``From the
  information bottleneck to the privacy funnel,'' in \emph{2014 IEEE
  Information Theory Workshop}, 2014, pp. 501--505.

\bibitem{sankar}
L.~Sankar, S.~R. Rajagopalan, and H.~V. Poor, ``Utility-privacy tradeoffs in
  databases: An information-theoretic approach,'' \emph{IEEE Transactions on
  Information Forensics and Security}, vol.~8, no.~6, pp. 838--852, 2013.

\bibitem{borz}
B.~{Rassouli} and D.~{G\"{u}nd\"{u}z}, ``On perfect privacy,'' \emph{IEEE
  Journal on Selected Areas in Information Theory}, vol.~2, no.~1, pp.
  177--191, 2021.

\bibitem{gun}
S.~{Sreekumar} and D.~{G\"{u}nd\"{u}z}, ``Optimal privacy-utility trade-off
  under a rate constraint,'' in \emph{2019 IEEE International Symposium on
  Information Theory}, July 2019, pp. 2159--2163.

\bibitem{khodam}
A.~Zamani, T.~J. Oechtering, and M.~Skoglund, ``A design framework for strongly
  $\chi^2$-private data disclosure,'' \emph{IEEE Transactions on Information
  Forensics and Security}, vol.~16, pp. 2312--2325, 2021.

\bibitem{Khodam22}
{A. Zamani, T. J. Oechtering, and M. Skoglund}, ``Data disclosure with non-zero
  leakage and non-invertible leakage matrix,'' \emph{IEEE Transactions on
  Information Forensics and Security}, vol.~17, pp. 165--179, 2022.

\bibitem{kostala}
Y.~Y. Shkel, R.~S. Blum, and H.~V. Poor, ``Secrecy by design with applications
  to privacy and compression,'' \emph{IEEE Transactions on Information Theory},
  vol.~67, no.~2, pp. 824--843, 2021.

\bibitem{dwork1}
C.~Dwork, F.~McSherry, K.~Nissim, and A.~Smith, ``Calibrating noise to
  sensitivity in private data analysis,'' in \emph{Theory of cryptography
  conference}.\hskip 1em plus 0.5em minus 0.4em\relax Springer, 2006, pp.
  265--284.

\bibitem{calmon4}
F.~P. {Calmon}, A.~{Makhdoumi}, M.~{Medard}, M.~{Varia}, M.~{Christiansen}, and
  K.~R. {Duffy}, ``Principal inertia components and applications,'' \emph{IEEE
  Transactions on Information Theory}, vol.~63, no.~8, pp. 5011--5038, Aug
  2017.

\bibitem{issajoon}
I.~Issa, A.~B. Wagner, and S.~Kamath, ``An operational approach to information
  leakage,'' \emph{IEEE Transactions on Information Theory}, vol.~66, no.~3,
  pp. 1625--1657, 2020.

\bibitem{asoo}
S.~Asoodeh, M.~Diaz, F.~Alajaji, and T.~Linder, ``Estimation efficiency under
  privacy constraints,'' \emph{IEEE Transactions on Information Theory},
  vol.~65, no.~3, pp. 1512--1534, 2019.

\bibitem{Total}
B.~Rassouli and D.~{G\"{u}nd\"{u}z}, ``Optimal utility-privacy trade-off with
  total variation distance as a privacy measure,'' \emph{IEEE Transactions on
  Information Forensics and Security}, vol.~15, pp. 594--603, 2020.

\bibitem{issa2}
I.~Issa, S.~Kamath, and A.~B. Wagner, ``Maximal leakage minimization for the
  shannon cipher system,'' in \emph{2016 IEEE International Symposium on
  Information Theory}, 2016, pp. 520--524.

\bibitem{zamani2022bounds}
A.~Zamani, T.~J. Oechtering, and M.~Skoglund, ``Bounds for privacy-utility
  trade-off with non-zero leakage,'' \emph{arXiv preprint arXiv:2201.08738},
  2022.

\bibitem{zamaniITW2022}
------, ``Bounds for privacy-utility trade-off with per-letter privacy
  constraints and non-zero leakage,'' \emph{arXiv preprint arXiv:2205.04881},
  2022.

\bibitem{kosenaz}
E.~Erdemir, P.~L. Dragotti, and D.~G{\"u}nd{\"u}z, ``Active privacy-utility
  trade-off against inference in time-series data sharing,'' \emph{arXiv
  preprint arXiv:2202.05833}, 2022.

\bibitem{naz2}
------, ``Privacy-aware communication over a wiretap channel with generative
  networks,'' in \emph{ICASSP 2022 - 2022 IEEE International Conference on
  Acoustics, Speech and Signal Processing (ICASSP)}, 2022, pp. 2989--2993.

\bibitem{9457633}
T.-Y. Liu and I.-H. Wang, ``Privacy-utility tradeoff with nonspecific tasks:
  Robust privatization and minimum leakage,'' in \emph{2020 IEEE Information
  Theory Workshop (ITW)}, 2021, pp. 1--5.

\bibitem{asoodeh1}
S.~Asoodeh, M.~Diaz, F.~Alajaji, and T.~Linder, ``Information extraction under
  privacy constraints,'' \emph{Information}, vol.~7, no.~1, p.~15, 2016.

\bibitem{MMSE}
S.~{Asoodeh}, F.~{Alajaji}, and T.~{Linder}, ``Privacy-aware mmse estimation,''
  in \emph{2016 IEEE International Symposium on Information Theory (ISIT)},
  2016, pp. 1989--1993.

\bibitem{nekouei2}
E.~Nekouei, T.~Tanaka, M.~Skoglund, and K.~H. Johansson,
  ``Information-theoretic approaches to privacy in estimation and control,''
  \emph{Annual Reviews in Control}, 2019.

\bibitem{bassi}
G.~Bassi, M.~Skoglund, and P.~Piantanida, ``Lossy communication subject to
  statistical parameter privacy,'' in \emph{2018 IEEE International Symposium
  on Information Theory (ISIT)}, 2018, pp. 1031--1035.

\bibitem{liu2020robust}
T.-Y. Liu, I.~Wang \emph{et~al.}, ``Robust privatization with non-specific
  tasks and the optimal privacy-utility tradeoff,'' \emph{arXiv preprint
  arXiv:2010.10081}, 2020.

\bibitem{warner1965randomized}
S.~L. Warner, ``Randomized response: A survey technique for eliminating evasive
  answer bias,'' \emph{Journal of the American Statistical Association},
  vol.~60, no. 309, pp. 63--69, 1965.

\bibitem{kosnane}
C.~T. Li and A.~El~Gamal, ``Strong functional representation lemma and
  applications to coding theorems,'' \emph{IEEE Transactions on Information
  Theory}, vol.~64, no.~11, pp. 6967--6978, 2018.

\end{thebibliography}


\begin{thebibliography}{10}
\providecommand{\url}[1]{#1}
\csname url@samestyle\endcsname
\providecommand{\newblock}{\relax}
\providecommand{\bibinfo}[2]{#2}
\providecommand{\BIBentrySTDinterwordspacing}{\spaceskip=0pt\relax}
\providecommand{\BIBentryALTinterwordstretchfactor}{4}
\providecommand{\BIBentryALTinterwordspacing}{\spaceskip=\fontdimen2\font plus
\BIBentryALTinterwordstretchfactor\fontdimen3\font minus
  \fontdimen4\font\relax}
\providecommand{\BIBforeignlanguage}[2]{{%
\expandafter\ifx\csname l@#1\endcsname\relax
\typeout{** WARNING: IEEEtran.bst: No hyphenation pattern has been}%
\typeout{** loaded for the language `#1'. Using the pattern for}%
\typeout{** the default language instead.}%
\else
\language=\csname l@#1\endcsname
\fi
#2}}
\providecommand{\BIBdecl}{\relax}
\BIBdecl

\bibitem{rassoul1}
B.~Rassouli and D.~G\"{u}nd\"{u}z, ``On perfect privacy and maximal
  correlation,'' \emph{arXiv preprint arXiv:1712.08500}, 2017.

\bibitem{makhdoumi}
A.~Makhdoumi, S.~Salamatian, N.~Fawaz, and M.~M{\'e}dard, ``From the
  information bottleneck to the privacy funnel,'' in \emph{2014 IEEE
  Information Theory Workshop (ITW 2014)}.\hskip 1em plus 0.5em minus
  0.4em\relax IEEE, 2014, pp. 501--505.

\bibitem{tishby}
N.~Tishby, F.~C. Pereira, and W.~Bialek, ``The information bottleneck method,''
  \emph{arXiv preprint physics/0004057}, 2000.

\bibitem{yamamoto}
H.~Yamamoto, ``A source coding problem for sources with additional outputs to
  keep secret from the receiver or wiretappers (corresp.),'' \emph{IEEE
  Transactions on Information Theory}, vol.~29, no.~6, pp. 918--923, 1983.

\bibitem{sankar}
L.~Sankar, S.~R. Rajagopalan, and H.~V. Poor, ``Utility-privacy tradeoffs in
  databases: An information-theoretic approach,'' \emph{IEEE Transactions on
  Information Forensics and Security}, vol.~8, no.~6, pp. 838--852, 2013.

\bibitem{dwork1}
C.~Dwork, F.~McSherry, K.~Nissim, and A.~Smith, ``Calibrating noise to
  sensitivity in private data analysis,'' in \emph{Theory of cryptography
  conference}.\hskip 1em plus 0.5em minus 0.4em\relax Springer, 2006, pp.
  265--284.

\bibitem{dwork2}
C.~Dwork, ``Differential privacy, in automata, languages and programming,''
  \emph{ser. Lecture Notes in Computer Scienc}, vol. 4052, p. 112, 2006.

\bibitem{oech}
Z.~Li, T.~J. Oechtering, and D.~G{\"u}nd{\"u}z, ``Privacy against a hypothesis
  testing adversary,'' \emph{IEEE Transactions on Information Forensics and
  Security}, vol.~14, no.~6, pp. 1567--1581, 2018.

\bibitem{borade}
S.~Borade and L.~Zheng, ``Euclidean information theory,'' in \emph{2008 IEEE
  International Zurich Seminar on Communications}.\hskip 1em plus 0.5em minus
  0.4em\relax IEEE, 2008, pp. 14--17.

\bibitem{huang}
S.-L. Huang and L.~Zheng, ``Linear information coupling problems,'' in
  \emph{2012 IEEE International Symposium on Information Theory
  Proceedings}.\hskip 1em plus 0.5em minus 0.4em\relax IEEE, 2012, pp.
  1029--1033.

\bibitem{huang2}
S.-L. Huang, C.~Suh, and L.~Zheng, ``Euclidean information theory of
  networks,'' \emph{IEEE Transactions on Information Theory}, vol.~61, no.~12,
  pp. 6795--6814, 2015.

\end{thebibliography}
\end{document}